\documentclass[a4paper,UKenglish]{lipics-v2016}

\usepackage{microtype}

\usepackage{hyperref}
\usepackage{amsfonts,amsmath,amsthm,wrapfig}
\usepackage[noend]{algorithmic}
\usepackage[boxed]{algorithm}
\usepackage{xspace}
\usepackage{framed}
\usepackage{xcolor}
\usepackage{verbatim}
\usepackage{tikz}
\usepackage{paralist}
\usetikzlibrary{arrows}
\allowdisplaybreaks

\title{New Tools and Connections for Exponential-time Approximation }
\titlerunning{New Tools and Connections for Exponential-time Approximation} 

\author[1]{Nikhil Bansal}
\author[2]{Parinya Chalermsook}
\author[3]{Bundit Laekhanukit}
\author[4]{Danupon Nanongkai}
\author[5]{Jesper Nederlof}
\affil[1]{Eindhoven University of Technology, The Netherlands.
	\texttt{n.bansal@tue.nl}.}
\affil[2]{Aalto University, Finland. \texttt{parinya.chalermsook@aalto.fi}.}
\affil[3]{Weizmann Institute of Science, Israel.
  \texttt{bundit.laekhanukit@weizmann.ac.il}.}
\affil[4]{KTH, Royal Institute of Technology, Sweden. \texttt{danupon@kth.se}}
\affil[5]{Eindhoven University of Technology, The Netherlands. \texttt{j.nederlof@tue.nl}. }

\authorrunning{N. Bansal et al.}

\Copyright{Nikhil Bansal, Parinya Chalermsook, Bundit Laekhanukit, Danupon Nanongkai, Jesper Nederlof}


\newcommand{\false}{\ensuremath{\mathbf{false}}\xspace}
\newcommand{\true}{\ensuremath{\mathbf{true}}\xspace}
\newcommand{\poly}{\ensuremath{\mathrm{poly}}\xspace}

\newcommand{\polyloglog}{\ensuremath{\mathrm{polyloglog}}\xspace}

\newcommand{\algorithmiccommentt}[1]{\colorbox{black!10}{#1}}

\newcommand{\eps}{\varepsilon}

\newcommand{\MIS}{{\sf Independent Set}\xspace}
\newcommand{\VC}{{\sf Vertex Cover}\xspace}
\newcommand{\COL}{{\sf Coloring}\xspace}
\newcommand{\SAT}{{\sf SAT}\xspace}

\newcommand{\ThreeSAT}{{\sf 3-CNF-SAT}\xspace}
\newcommand{\CSP}{{\sf CSP}\xspace}
\newcommand{\eset}{{\mathcal E}}

\def\ShowComment{True}

\ifdefined\ShowComment

\newcounter{note}[section]

\def\danupon#1{\marginpar{$\leftarrow$\fbox{D}}\footnote{$\Rightarrow$~{\sf #1 --Danupon}}}

\def\parinya#1{\marginpar{$\leftarrow$\fbox{P}}\footnote{$\Rightarrow$~{\sf #1 --Parinya}}}

\else

\def\danupon#1{}
\def\parinya#1{}

\fi

\begin{document}

\maketitle

\begin{abstract}
	In this paper, we develop new tools and connections for {\em exponential time approximation}. 
	In this setting, we are given a problem instance and a parameter $\alpha>1$, and the goal is to design an $\alpha$-approximation algorithm with the fastest possible running time. We show the following results: 
	\begin{enumerate}
		\item An $r$-approximation for  maximum independent set in $O^*(\exp(\tilde O(n/r \log^2 r+r\log^2r)))$ time, 
		\item An $r$-approximation for  chromatic number in $O^*(\exp(\tilde{O}(n/r \log r+r\log^2r)))$  time,
		\item A $(2-1/r)$-approximation for  minimum vertex cover  in $O^*(\exp(n/r^{\Omega(r)}))$  time, and
		\item A $(k-1/r)$-approximation for minimum $k$-hypergraph vertex cover  in $O^*(\exp(n/ (kr)^{\Omega(kr)}))$ time. 
	\end{enumerate} 
	(Throughout, $\tilde O$ and $O^*$ omit $\polyloglog(r)$ and factors polynomial in the input size, respectively.) The best known time bounds for all problems were   $O^*(2^{n/r})$ [Bourgeois et al. 2009, 2011 \& Cygan et al. 2008].
	For  maximum independent set and chromatic number, these bounds were complemented by  $\exp(n^{1-o(1)}/r^{1+o(1)})$ lower bounds (under the Exponential Time Hypothesis (ETH))~[Chalermsook et al., 2013 \& Laekhanukit, 2014 (Ph.D. Thesis)]. Our results show that the naturally-looking $O^*(2^{n/r})$ bounds are not tight for all these problems. 
	The key to these algorithmic results is a  {\em sparsification} procedure that reduces a problem to its bounded-degree variant, allowing the use of better approximation algorithms for bounded degree graphs. For obtaining the first two results, we introduce a new {\em randomized branching rule}.
	
	Finally, we show a connection between PCP parameters and  exponential-time approximation algorithms. This connection together with our independent set algorithm refute the possibility to overly reduce the size of Chan's PCP [Chan, 2016]. 
	It also implies that a (significant) improvement over our result will refute the gap-ETH conjecture [Dinur 2016 \& Manurangsi and Raghavendra, 2016].  	
\end{abstract}

\section{Introduction}
\label{sec:intro}

The \MIS, \VC, and \COL problems are central problems in combinatorial optimization and have been extensively studied.
Most of the classical results concern either approximation algorithms that run in polynomial time or exact algorithms that run in (sub)exponential-time. While these algorithms are useful in most scenarios, they lack flexibility: Sometimes, we wish for a better approximation ratio with worse running time (e.g. computationally powerful devices), or faster algorithms with less accuracy. 
In particular, the trade-off between the running time and approximation ratios are needed in these settings.

Algorithmic results on the trade-offs between approximation ratio have been studied already in the literature in several settings, most notably in the context of {\em Polynomial-time Approximation Schemes (PTAS)}.   
For instance, in planar graphs, Baker's celebrated approximation scheme for several NP-hard problems~\cite{DBLP:journals/jacm/Baker94} gives an $(1+\eps)$-approximation for e.g.~\MIS in time $O^*(\exp(O(1/\eps)))$ time. In graphs of small treewidth, Czumaj et al.~\cite{CHLN05} give an $O^*(\exp(tw/r))$ time algorithm that given a graph along with a tree decomposition of it of width at most $tw$, find an $r$-approximation for \MIS. For general graphs, approximation results for several problems have been studied in several works (see e.g.~\cite{BLP16,BP16,BEP11,CyganP10,CyganKW09,CKPW08}). 
A basic building block that lies behind many of these results is to partition the input instance in smaller parts in which the optimal (sub)solution can be computed quickly (or at least faster than fully exponential-time). For example, to obtain an $r$-approximation for \MIS one may arbitrarily partition the vertex set in $r$ blocks and restrict attention to independent sets that are subsets of these blocks to get a $O^*(\exp(n/r))$ time $r$-approximation algorithm.

While at first sight one might think that such a na\"ive algorithm should be easily improvable via more advanced techniques, it was shown in~\cite{CLN13,BLP16} that almost linear-size PCPs~\cite{Dinur07,MoshkovitzR10} imply that $r$-approximating \MIS \cite{CLN13} and \COL \cite{Laekhanukit2014-Thesis} requires at least $\exp(n^{1-o(1)}/r^{1+o(1)})$ time assuming the popular Exponential Time Hypothesis (ETH). In the setting of the more sophisticated Baker-style approximation schemes for planar graphs, Marx~\cite{Marx07} showed that no $(1+\eps)$-approximating algorithm for planar \MIS can run in time $O^*(\exp((1/\eps)^{1-\delta}))$ assuming ETH, which implies that the algorithm of Czumaj cannot be improved to run in time $O^*(\exp(tw/r^{1+\eps}))$. 

These lower bounds, despite being interesting, are far from tight and by no means answer the question whether the known approximation trade-offs can be improved significantly, and in fact in many settings we are far from understanding the full power of exponential time approximation. For example we cannot exclude algorithms that $2$-approximate $k$-\MIS\footnote{That is, given a graph and integer $k$ answer YES if it has an independent set of size at least $2k$ and NO if it has no independent set of size at least $k$.} in time $O^*(f(k))$ for some function $f$ (see e.g.~\cite{KS16}), nor do we know algorithms that run asymptotically faster than the fastest exact algorithm that runs in time $n^{0.792k}$ time~\cite{NP85}. 

In this paper we aim to advance this understanding and study the question of designing as fast as possible approximation algorithms for \MIS, \COL and \VC in general (hyper)graphs.

\subsection*{Our Results.}

For \MIS our result is the following. Here we use $\tilde{O}$ to omit log log factors in $r$.

\begin{theorem}\label{thm:main}
	There is a randomized algorithm that given an $n$-vertex graph $G$ and integer $r$ outputs an independent set that, with constant probability, has size at least $\alpha(G)/r$, where $\alpha(G)$ denotes the maximum independent set size of $G$. The algorithm runs in time $O^*(\exp(\tilde{O}(n/(r\log^2 r)+r \log^2 r)))$.
\end{theorem}

To prove this result we introduce a new \emph{randomized branching rule} that we will now introduce and put in context towards previous results. This follows a \emph{sparsification technique} that reduces the maximum degree to a given number. This technique was already studied before in the setting of exponential time approximation algorithms \MIS by Cygan et al. (see~\cite[paragraph `Search Tree Techniques']{CKPW08}) and Bourgeois et al. (see~\cite[Section 2.1]{BEP11}), but the authors did not obtain running times sub-exponential in $n/r$. Specifically, the sparsification technique is to branch (e.g. select a vertex and try to both include $v$ in an independent set or discard and recurse for both possibilities) on vertices of sufficiently high degree. The key property is that if we decide to include a vertex and the independent set, we may discard all neighbors of $v$. If we generate instances by keep branching on vertices of degree at least $d$ until the maximum degree is smaller than $d$, then at most $\binom{n}{n/d} \lessapprox \exp(n\log(d)/d)$ instances are created. In each such instance, the maximum independent set can be easily $d$-approximated by a greedy argument. Cygan et al.~\cite{CKPW08} note that this gives worse than $O^*(2^{n/r})$ running times. 

Our algorithm works along this line but incorporates two (simple) ideas. 
Our first observation is that instead of solving each leaf instance by greedy $d$-approximation algorithm, one can use a recent $\tilde O(\frac{d}{\log^2 d})$ approximation algorithm by Bansal et al.~\cite{BGG15} for \MIS on bounded degree graphs. 
If we choose $d \approx r \log^2 r$, this immediately gives an improvement, an $r$-approximation in time essentially ${\sf exp}(\frac{n}{r \log r})$. 
To improve this further we use present an additional (more innovative) idea introducing randomization. This idea relies on the fact that in the sparsification step we have (unexploited) slack as we aim for an approximation.\footnote{This observation was already made by Bourgeois et al.~\cite{BEP11}, but we exploit it in a new way.} Specifically, whenever we branch, we only consider the `include' branch with probability $1/r$. This will lower the expected number of produced leaf instances in the sparsification step to $2^{n/d} \approx {\sf exp}(\frac{n}{r \log^2 r})$ and preserves the approximation factor with good probability.

Via fairly standard methods (see e.g.~\cite{BHK09}) we show this also gives a faster algorithm for coloring in the following sense:

\begin{theorem}\label{thm:col}
	There is a randomized algorithm that, given an $n$-vertex graph $G$ and an integer $r>0$, outputs with constant probability a proper coloring of $G$ using at most $r\cdot\chi(G)$ colors. The algorithm runs in time $O^*(\exp(\tilde{O}(n/(r\log r)+r \log^2 r)))$.
\end{theorem}

As a final indication that sparsification is a very powerful tool to obtain fast exponential time approximation algorithms, we show that a combination of a result of Halperin~\cite{Halperin02} and the sparsification Lemma~\cite{IPZ01} gives the following result for the Vertex Cover problem in hypergraphs with edges of size at most $k$ (or Set Cover problem with frequency at most $k$). 

\begin{theorem}\label{apx:vchyp}
	For every $k$, there is an $r_0:=r(k)$ such that for every $r \geq r_0$ there is an $O^*(\exp(\frac{n}{(kr)^{\Omega(kr)}}))$ time $(k-\tfrac{1}{r})$-approximation algorithm for the Vertex Cover problem in hypergraphs with edges of size at most $k$.
\end{theorem}

Note that for $k=2$ (e.g. vertex cover in graphs), this gives an $O^*(\exp(\frac{n}{r^{\Omega(r)}}))$ running time, which gives an exponential improvement (in the denominator of the exponent)  upon the $(2-1/r)$ approximation by Bonnet et al.~\cite{BEP11} that runs in time $O^*(2^{n/r})$. It was recently brought to our attention that Williams and Yu~\cite{WY} independently have unpublished results for (hypergraph) vertex cover and independent set using sparsification techniques similar to ours.

\subparagraph*{Connections to PCP parameters}
The question of approximating the maximum independent set problem in sub-exponential time has close connections to the trade-off  between three important parameters of PCPs: {\em size}, {\em gap} and {\em free-bit}.
We discuss the implications of our algorithmic results in terms of these PCP parameters. 

Roughly speaking, the gap parameter is the ratio between completeness and soundness, while the {\em freeness} parameter is the number of distinct proofs that would cause the verifier to accept; the {\em free-bit} is simply a logarithm of freeness.
For convenience, we will continue our discussions in terms of freeness, instead of freebit.

\begin{itemize}
    \item {\bf Freebit v.s. gap:} The dependency between freeness and gap has played important role in hardness of approximation.
Most notably,  the existence of PCPs with freeness $g^{o(1)}$ where $g$ is a gap parameter is ``equivalent'' to $n^{1-o(1)}$ hardness of approximating maximum independent set~\cite{Hastad96,BellareGS98}; this result is a building block for proving other hardness of approximation for many other combinatorial problems, e.g., coloring~\cite{FeigeK98}, disjoint paths, induced matching, cycle packing, and pricing.
So it is fair to say that this PCP parameter trade-off captures the approximability of many natural combinatorial problems.

Better parameter trade-off implies stronger hardness results.
The existence of a PCP with arbitrarily large gap and freeness $1$ (lowest possible) is in fact equivalent to $(2-\epsilon)$ inapproximability for \VC.
The best known trade-off is due to Chan~\cite{Chan16}: For any $g >0$, there is a polynomial-sized PCP with gap $g$ and freeness $O(\log g)$, yielding the best known NP-hardness of approximating maximum independent set in sparse graphs, i.e. $\Omega(d/ \log^4 d)$ NP-hardness of approximating maximum independent set in degree-$d$ graphs.~\footnote{Roughly speaking, the existence of a PCP with freeness $F(g)$ (where $g$ is a gap) implies $\Omega(\frac{d}{F(d) \log^3 d})$ hardness of approximating independent set in degree-$d$ graphs.}

\item {\bf Size, freebit, and gap:} When a polynomial-time approximation algorithm is the main concern, polynomial size PCPs are the only thing that matter.
But when it comes to exponential time approximability, another important parameter, {\em size} of the PCPs, has come into play.
The trade-off between size, freebit, and gap tightly captures the (sub-)exponential time approximability of many combinatorial problems.
For instance, for any $g >0$, Moshkovitz and Raz~\cite{MoshkovitzR10} constructs a PCP of size $n^{1+o(1)}$ and freeness $2^{O(\sqrt{\log g})}$; this implies that $r$-approximating \MIS requires time $2^{n^{1-o(1)}/r^{1+o(1)}}$~\cite{CLN13}.

\end{itemize}

Our exponential-time approximation result for \MIS implies the following tradeoff results.

\begin{corollary}
Unless ETH breaks, a freebit PCP with gap $g$, freeness $F$ and size $S$ must satisfy $F \cdot S = \Omega(n \log^2 g)$.
\end{corollary}

In particular, this  implies that (i) Chan's PCP cannot be made smaller size than $o(n \log g)$, unless ETH breaks, and (ii) in light of the equivalence between gap-amplifying freebit PCPs with freeness $1$ and $(2-\epsilon)$ approximation for \VC, our result shows that such a PCP must have size at least $\Omega(n \log^2 g)$.
We remark that no such trade-off results are known for polynomial-sized PCPs. To our knowledge, this is the first result of its kind.

\subparagraph*{Further related results}

The best known results for \MIS in the polynomial-time regime are an $O(\frac{n (\log \log n)^2}{\log^3 n})$-approximation~\cite{Feige04}, and the hardness of $n/{\sf exp}(O(\log^{3/4+ o(1)} n))$ (which also holds for \COL)~\cite{KhotP06}. 
For \VC, the best known hardness of approximation is $(\sqrt{2} - o(1))$ NP-hardness~\cite{Khot16} and $(2-\epsilon)$ hardness assuming the unique games conjecture~\cite{KhotR08}. 
All three problems (\MIS, \COL, and \VC) do not admit exact algorithms that run in time $2^{o(n)}$, unless ETH fails. Besides the aforementioned works~\cite{BEP11,CKPW08} sparsification techniques for exponential time approximation were studied by Bonnet and Paschos in~\cite{BP16}, but mainly hardness results were obtained.
 
\section{Preliminaries}
\label{sec:prelim}

We first formally define the three problems that we consider in this paper. \MIS: Given a graph $G=(V,E)$, we say that $J \subseteq V$ is an independent set if there is no edge with both endpoints in $J$. The goal of \MIS is to output an independent set $J$ of maximum cardinality. Denote by $\alpha(G)$, the cardinality of the maximum independent set. \VC: Given a graph $G=(V,E)$, we say that $J \subseteq V$ is a vertex cover of $G$ if every edge is incident to at least one vertex in $J$. The goal of \VC is to output a vertex cover of minimum size. A generalization of vertex cover, called $k$-Hypergraph Vertex Cover $k$-\VC, is defined as follows. Given a hypergraph $G=(V,\eset)$ where each hyperedge $h \in \eset$ has cardinality at most $k$, the goal is to find a collection of vertices $J \subseteq V$ such that each hyperedge is incident to at least one vertex in $J$, while minimizing $|J|$. The \emph{degree $\Delta(H)$} of hypergraph $H$ is the maximum frequency of an element. \COL: Given a graph $G=(V,E)$, a proper $k$-coloring of $G$ is a function $f: V \rightarrow [k]$ such that $f(u) \neq f(v)$ for all $uv \in E$. The goal of \COL is to compute a minimum integer $k>0$ such that $G$ admits a (proper) $k$-coloring; this number is referred to as the {\em chromatic number}, denote $\chi(G)$.

For a graph $G=(V,E)$, $N_G(v)$ denotes the set of neighbors of $v$ and $d_G(v)$ denotes $|N_G(v)|$. If $X \subseteq V$ we let $G[X]$ denote the graph $(X, E \cap (X \times X))$ i.e. the subgraph of $G$ induced by $X$
We use $\exp(x)$ to denote $2^x$ in order to avoid superscripts.
We use the $O^*(\cdot)$-notation to suppress factors polynomial in the input size. We use $\tilde{O}$ and $\tilde{\Omega}$ to suppress factors polyloglog in $r$ in respectively upper and lower bounds and write $\tilde{\Theta}$ for all functions that are in both $\tilde{O}$ and $\tilde{\Omega}$.

\section{Faster Approximation via Randomized Branching and Sparsification}
\label{sec:approx-via-sparsifier}

\subsection{Maximum Independent Set}
\label{sec:independent-set}

In this section we prove Theorem~\ref{thm:main}. Below is our key lemma.

\begin{lemma}\label{lem:red}
	Suppose there is an approximation algorithm $\mathtt{dIS}(G,r)$ that runs in time $T(n,r)$ and outputs an Independent Set of $G$ of size $\alpha(G)/r$ if $G$ has maximum degree $d(r)$, (where $d(r) \geq 2r$). Then there is an algorithm $\mathtt{IS}(G,r)$ running in expected time $O^*\left(\exp\left(\tfrac{n}{d(r)}\log(4d(r)/r)\right)T(n,r)\right)$ that outputs an independent set of expected size $\alpha(G)/r$.
\end{lemma}
\begin{proof}
	Consider the following algorithm.
	\begin{figure}[H]
		\begin{framed}
			\begin{algorithmic}[1]
				\REQUIRE $\mathtt{IS}(G=(V,E),r)$
				\IF{$\exists v \in V: d_G(v) \geq d(r)$} \label{lin:if}
				\STATE Draw a random Boolean variable $b$ such that $\Pr[b=\true]=1/r$.
				\IF{$b=\true$}
				\STATE \algorithmicreturn\ the largest of $IS(G[V\setminus v])$ and $IS(G[V \setminus N(v)]) \cup \{v\}$. \label{lin:branch}
				\ELSE
				\STATE \algorithmicreturn\ $IS(G[V\setminus v])$.
				\ENDIF
				\ELSE
				\STATE \algorithmicreturn\ $dIS(G)$.\label{lin:degapx}
				\ENDIF
			\end{algorithmic}
		\end{framed}
		\caption{Approximation algorithm for Independent Set using an approximation algorithm $dIS$ that works in bounded degree graphs.}
		\label{alg:IS}
	\end{figure}

	For convenience, let us fix $r$ and $d:=d(r)$.
	We start by analyzing the expected running time of this algorithm. Per recursive call the algorithm clearly uses $O^*(T(n,r))$ time. It remains to bound the number of recursive calls $R(n)$ made by $\mathtt{IS}(G,r)$ when $G$ has $n$ vertices. We will bound $R(n) \leq 2^{\lambda n}$ for $\lambda= \log(4d/r)/d$ by induction on $n$. Note that here $\lambda$ is chosen such that
	\begin{equation}\label{eq:lambda}
	\exp(-\lambda\cdot d) = r/(4d) \leq \frac{r\log(4d/r)}{2d},
	\end{equation}
	where we use $d/r \geq 2$ for the inequality. For the base case of the induction, note that if the condition at Line~\ref{lin:if} does not hold, the algorithm does not use any recursive calls and the statement is trivial as $\lambda$ is clearly positive. For the inductive step, we see that
	\begin{align*}
	R(n) &\leq R(n-1) + \Pr[b=\true]\cdot R(n-d)&\\
	&= R(n-1) + R(n-d)/r&\\
	&= \exp(\lambda(n-1))+\exp(\lambda(n-d))/r&\\
	&= \exp(\lambda n)\left( \exp(-\lambda)+\exp(-\lambda d)/r\right)& \hfill\algorithmiccommentt{Using $\exp(-x) \leq 1-x/2$ for $x \in [0,1]$}\\
	&\leq \exp(\lambda n)\left( 1-\lambda/2+\exp(-\lambda d)/r\right)& \hfill\algorithmiccommentt{Using $\exp(-\lambda\cdot d(r)) \leq \lambda r/2$ from~\eqref{eq:lambda}} \\
	&\leq \exp(\lambda n).
	\end{align*}
	We continue by analyzing the output of the algorithm. It clearly returns a valid independent set as all neighbors of $v$ are discarded when $v$ is included in Line~\ref{lin:branch} and an independent set is returned at Line~\ref{lin:degapx}. It remains to show $\mathbb{E}[|\mathtt{IS}(G,r)|] \geq \alpha(G)/r$ which we do by induction on $n$. In the base case in which no recursive call is made, note that on Line~\ref{lin:degapx} we indeed obtain an $r$-approximation as $G$ has maximum degree $d(r)$. For the inductive case, let $X$ be a maximum independent set of $G$ and let $v$ be the vertex as picked on Line~\ref{lin:if}. We distinguish two cases based on whether $v \in X$. If $v \notin X$, then $\alpha(G)=\alpha(G[V\setminus v])$ and the inductive step follows as $\mathbb{E}[|\mathtt{IS}(G[V \setminus v],r)|] \geq \alpha(G)/r$ by the induction hypothesis. Otherwise, if $v \in X$, then $\mathbb{E}[|\mathtt{IS}(G,r)|]$ is at least
	\begin{align*}
	 &\Pr[b=\false]\cdot\mathbb{E}[|\mathtt{IS}(G[V \setminus \{v\}],r)|]+\Pr[b=\true]\cdot\mathbb{E}[|\mathtt{IS}(G[N \setminus N(v)],r)|+1]\\
	\geq\ &\left(1-\tfrac{1}{r}\right)\frac{\alpha(G)-1}{r} + \tfrac{1}{r}\left(\frac{\alpha(G)-1}{r}+1\right) \\
	=\ &\frac{\alpha(G)-1}{r}+\tfrac{1}{r} = \alpha(G)/r,
	\end{align*}
	as required. Here the first inequality uses the induction hypothesis twice.
\end{proof}

We will invoke the above lemma by using the algorithm $dIS(G)$ by Bansal et al.~\cite{BGG15} implied by the following theorem:

\begin{theorem}[\cite{BGG15}, Theorem 1.3]\label{thm:bansalind}
	There is an $\tilde{O}(d/\log^2 d)$ approximation algorithm $dIS(G)$ for Independent Set on graphs of maximum degree $d$ running in time $O^*(\exp(O(d)))$.
\end{theorem}

\begin{proof}[Proof of Theorem~\ref{thm:main}]
	We may apply Lemma~\ref{lem:red} with $r/3$ and, by virtue of Theorem~\ref{thm:bansalind}, with $d(r/3)=\tilde{\Theta}(r\log^2 r)$, and $T(n,r)=O^*(\exp(\tilde{O}(r\log^2 r)))$. We obtain an $O^*(\exp(\tilde{O}(n/r\log^2 r + r\log^2r)))$ expected time algorithm that outputs an independent set of expected size $2\alpha(G)/r$.

	Since the size of the output is upper bounded by $\alpha(G)$ we obtain an independent set of size at least $\alpha(G)/r$ with probability at least $1/(3r)$, and we may boost this probability to $3/4$ by $O(r)$ repetitions.
	
	By Markov's inequality these repetitions together run in $O^*(\exp(\tilde{O}(n/r\log^2 r + r\log^2r)))$ time with probability $3/4$. The theorem statement follows by a union bound as these $O(r)$ repetitions run in the claimed running time and simultaneously some repetition finds an independent set of size at least $\alpha(G)/r$, with probability at least $1/2$.	
\end{proof}

\subparagraph*{A deterministic algorithm:} Additionally, we also show a deterministic $r$-approximation algorithm that runs in time $\exp(\tilde{O}(n/r \log r))$.
The algorithm utilizes Feige's algorithm~\cite{Feige04} as a blackbox, and is deferred to Appendix~\ref{app:feige}.

\subsection{Graph Coloring}
\label{sec:graph-coloring}

Now we use the approximation algorithm for \MIS as a subroutine for an approximation algorithm for \COL to prove Theorem~\ref{thm:col} as follows: 

\begin{proof}[Proof of Theorem~\ref{thm:col}]
	The algorithm combines the approximation algorithm $\mathtt{IS}$ from Section~\ref{sec:independent-set} for \MIS with an exact algorithm $\mathtt{optcol}$ for \COL (see, e.g.,~\cite{BHK09}) as follows:
	\begin{figure}[H]
		\begin{framed}
			\begin{algorithmic}[1]
				\REQUIRE $\mathtt{CHR}(G=(V,E),r)$
				\STATE Let $n=|V|$, $c=0$. \label{chr:first-line}
				\WHILE{$|V| \geq n/(r \log r)$ } \label{lin:while}
				\STATE $c \gets c+1$.
				\STATE $C_c \gets \mathtt{IS}(G[V],r/\ln(r \log r))$. \label{lin:invapxis}
				\STATE $V \gets V \setminus C_c$.
				\ENDWHILE \label{chr:end-first-phase}
				\STATE Let $(C_{c+1},\ldots,C_{\ell}) \gets \mathtt{optcol}(G[V])$ be some optimum coloring of the remaining graph $G(V)$. 
				\label{lin:optcol}
				\STATE \algorithmicreturn\ $(C_1,\ldots,C_{\ell})$.
			\end{algorithmic}
		\end{framed}
		\caption{Approximation algorithm for the chromatic number.}
		\label{alg:chr}
	\end{figure}
	
	We claim that $\mathtt{CHR}(G,r)$ returns with high probability a proper coloring of $G$ using $\ell \leq (r+2)\cdot \chi(G)$ colors. To prove the theorem, we invoke $\mathtt{CHR}(G,r-2)$ which has the same asymptotic running time. First, note that in each iteration of the while loop (Line~\ref{lin:while} of Algorithm~\ref{alg:chr}), $|V|$ is decreased by a multiplicative factor of at most $1-\frac{\ln(r \log r)}{r\cdot \chi(G)}$ because $G[V]$ must have an independent set of size at least $n/\chi(G)$ and therefore $|C_c| \geq \ln(r \log r) n/(r\cdot \chi(G))$. Before the last iteration, we have $|V| \geq n / (r \ln r)$. Thus, the number $\ell$ of iterations must satisfy
	$$
	1/(r \log r) \leq \left( 1-\frac{\ln(r \log r)}{r\cdot\chi(G)} \right)^{\ell -1} \leq \exp\left(-\frac{\ln(r \log r)(\ell-1)}{r\cdot\chi(G)}\right).
	$$
	This implies that $(\ell-1) \leq r\cdot\chi(G)$. Consequently, the number of colors used in the first phase of the algorithm (Line~\ref{chr:first-line} to Line~\ref{chr:end-first-phase}) is $c \leq r\chi(G)+1$. The claimed upper bound on $\ell$ follows because the number of colors used  for $G[V]$  in the second phase (Line~\ref{lin:optcol}) is clearly upper bounded by $\chi(G)$.
	
	To upper bound the running time, note that Line~\ref{lin:invapxis} runs in time
	$$
	\exp\left(\tilde{O}\left(\frac{n\ln(r \log r)}{r \log^2 (r/\ln(r \log r) )}+ r\log^2 r \right) \right) = \exp\left(\tilde{O}\left(\frac{n}{r \log r}\right)+r\lg^2 r\right),
	$$
	and implementing $\mathtt{optcol}(G=(V,E))$ by using the $O^*(2^{|V|})$ time algorithm from~\cite{BHK09}, Line~\ref{lin:optcol} also takes $O^*(2^{n/(r \log r)})$ time and the running time follows.
\end{proof}

\subsection{Vertex Cover and Hypergraph Vertex Cover}
\label{sec:vertex-cover}

In this section, we show an application of the sparsification technique to \VC to obtain Theorem~\ref{apx:vchyp}.
Here the sparsification step is not applied explicitly. Instead, we utilize the sparsification Lemma of Impagliazzo et al.~\cite{IPZ01} as a blackbox.
Subsequently, we solve each low-degree instance by using an algorithm of Halperin~\cite{Halperin02}.
The sparsification lemma due to Impagliazzo et al.~\cite{IPZ01}, shows that an instance of the $k$-Hypergraph Vertex Cover problem can be reduced to a (sub-)exponential number of low-degree instances.\footnote{The original formulation is for the Set Cover problem and the most popular formulation is for CNF-SAT problem, but they are all equivalent by direct transformation.}

\begin{lemma}[Sparsification Lemma, \cite{IPZ01,CIP06}]\label{lem:sparsify}
	There is an algorithm that, given a hypergraph $H=(V,\eset)$ with edges of size at most $k \geq 2$, a real number $\eps > 0$, produces set systems $H_1=(V,\eset_1),\ldots,H_\ell=(V,\eset_\ell)$ with edges of size at most $k$ in $O^*(\ell)$ time such that
	\begin{enumerate}
		\item every subset $X \subseteq V$ is a vertex cover of $H$ if and only if $X$ is a vertex cover of $H_i$ for some $i$,
		\item for every $i=1,\ldots,\ell$, the degree $\Delta(H_i)$ is at most $(k/\eps)^{3k}$,
		\item $\ell$ is at most $\exp(\eps n)$.
	\end{enumerate}
\end{lemma}

The next tool is an approximation algorithm for the $k$-Hypergraph Vertex Cover problem when the input graph has low degree due to Halperin~\cite{Halperin02}. 

\begin{theorem}[\cite{Halperin02}]\label{thm:halperin}
	There is a polynomial time $k-(1-o(1))\frac{k(k-1)\ln \ln \Delta}{\ln \Delta}$-approximation algorithm for the vertex cover problem in hypergraphs with edges of size at most $k$ in which every element has degree at most $\Delta$, for large enough $\Delta:=\Delta(k)$.
\end{theorem}

Now we complete the proof of the theorem by applying Lemma~\ref{lem:sparsify} with parameter $\eps= k/(kr)^{kr}$.
The number of low-degree instances $H_i$ produced by Lemma~\ref{lem:sparsify} is at most $\exp(\eps n) = \exp\left(O\left(\frac{k}{(kr)^{kr}}\right)\right)$.
Each graph $H_i$ has degree at most $\Delta(H_i) \leq (k/\eps)^{3k} = (kr)^{3k^2 r}$.
Note that
\[
\frac{\ln \ln \Delta(H_i)}{\ln \Delta(H_i)} \geq \frac{\ln(3k^2 r \ln (kr))}{3k^2 r \ln (kr)} \geq \frac{1}{3k^2r}.
\]

Plugging this value of $\Delta(H_i)$, Halperin's algorithm gives the approximation factor of
$$
k - \frac{k(k-1)\ln \ln \Delta}{\ln \Delta} \leq k- \frac{1}{6r}.
$$

Thus this gives an $k-1/(6r)$-approximation running in time $O^*(\exp(nk/(kr)^{kr}))$ which translates to an $k-1/r$-approximation running in time $O^*(\exp(nk/(kr/6)^{kr/6}))$.

\newcommand{\fgpcp}{\mbox{\sf FGPCP}\xspace}
\newcommand{\val}{{\sf val}}
\section{PCP Parameters and Exponential-time approximation hardness}

Exponential-time approximation has connections to the trade-off questions between three parameters of PCPs: {\em size}, {\em freebit}, and {\em gap}.
To formally quantify this connection, we define new terms, formally illustrating the ideas that have been already around in the literature.
We define a class of languages \fgpcp which stands for {\em Freebit and Gap-amplifiable PCP}.
Let $g$ be a positive real, and $S, F$ be non-decreasing functions.
A language $L$ is in $\fgpcp_{c}(S, F)$ if there is a constant $g_0 >1$ such that, for all constants $g \geq g_0$, there is a verifier $V_{g}$ that, on input $x \in \{0,1\}^n$, has access to a proof $\pi: |\pi| = O(S(n,g))$ and satisfies the properties:
\begin{itemize}
\item The verifier runs in $2^{o(n)}$ time.
\item If $x \in L$, then there is a proof $\pi$ such that $V^{\pi}_g(x)$ accepts with probability $\geq c$.
\item If $x \not\in L$, then for any proof $\pi$, $V^{\pi}_g(x)$ accepts with probability $\leq c/g$.
\item For each $x$ and each random string $r$, the verifier has $\leq F(g)$ accepting configurations.
\end{itemize}

The parameters $g$, $S$ and $\log F$ are referred to as {\em gap}, {\em size} and {\em freebit} of the PCPs respectively. For convenience, we call $F(g)$ the {\em freeness} of the PCP.
An intuitive way to view this PCP is as a class of PCPs parameterized by gap $g$.
An interesting question in the PCPs and hardness of approximation literature has been to find the smallest functions $S$ and $F$.

\begin{theorem}\label{thm:pcp and sub-expo}
If $\SAT \in \fgpcp_\delta(S, F)$ for some function $S(n,g)$ that is at least linearly growing in $n$, then for any constant $r$,  $r$-approximating \MIS, in input graph $G$, cannot be done in time $2^{o(S^{-1}(|V(G)|,r)/r F(r))}$ unless ETH fails. (we think of $r$ as a fixed number, and therefore $S(n,r)$ should be seen as a function on a single variable $n$.)
\end{theorem}

We prove the theorem later in this section.

\begin{corollary}
\label{cor:tradeoff}
Assuming that \SAT has no $2^{o(n)}$-time randomized algorithm and that $\SAT \in \fgpcp_{\delta}(S,F)$, then it must be the case that $S(n,g) \cdot F(g) = \Omega(n \cdot \frac{\log^2 g}{{\sf poly}( \log \log g)})$.
\end{corollary}

\begin{proof}
Otherwise, $S^{-1}(|V(G)|, r) = o(|V(G)| \cdot \frac{F(r) {\sf poly}( \log \log r)}{\log^2 r})$, and the Theorem~\ref{thm:pcp and sub-expo} would imply that there is no $2^{o(|V(G)| \cdot \frac{{\sf poly}( \log \log r)}{r \log^2 r})}$, contradicting the existence of our \MIS approximation algorithm.
\end{proof}

Now let us phrase the known PCPs in our framework of \fgpcp.  Chan's PCPs \cite{Chan16} can be stated that $\SAT \in \fgpcp_{1-o(1)}(\poly, O(\log g))$. Applying our results, this means that if one wants to keep the same freebit parameters given by Chan's PCPs, then the size must be at least $\Omega(n \log g)$.
Another interesting consequence is a connection between \VC and Freebit PCPs in the polynomial time setting~\cite{BellareGS98}.

\begin{theorem}[\cite{BellareGS98}]
\label{thm:vc-hardness}
\VC is $(2-\epsilon)$ hard to approximate if and only if $\SAT \in {\sf FGPCP}_{1/2 -\epsilon}({\sf poly}, 1)$.
\end{theorem}

The intended PCPs in Theorem~\ref{thm:vc-hardness} have arbitrary small soundness while the freeness remains $1$.
Our Corollary~\ref{cor:tradeoff} implies that such a PCP must have size at least $\Omega(n \log^2 g)$.

\subsection{Proof of Theorem~\ref{thm:pcp and sub-expo}}

\subparagraph*{Step 1: Creating a hard CSP}
We will need the following lemma that creates a ``hard'' \CSP from  \fgpcp.
This \CSP will be used later to construct a hard instance of \MIS.

\begin{lemma}
\label{lem: pcp to csp}
If $\SAT \in \fgpcp_{\delta}(S,F)$, then, for any $g>1$, there is a randomized reduction from an $n$-variable \SAT $\phi$ to a \CSP $\phi'$ having the following properties (w.h.p.):
\begin{itemize}
\item The number of variables of $\phi'$ is $\leq S(n)$.

\item The number of clauses of $\phi'$ is $\leq 10 S(n) g/\delta$.

\item The freeness of $\phi'$ is $\leq F(g)$.

\item If $\phi$ is satisfiable, then $\val(\phi') \geq \delta/2$.
      Otherwise, $\val(\phi') \leq 6\delta/g$.
\end{itemize}
\end{lemma}

\begin{proof}
Let $g$ be any number and $V_g$ be the corresponding verifier. 
On   input $\phi$, we create a \CSP $\phi'$ as follows.
For each proof bit $\Pi_i$, we have  variable $x_i$.
The set of variables is $X = \{x_1,\ldots, x_{S(n)}\}$.
We perform $M = 10 \lceil S(n)g/\delta \rceil$ iterations.
In iteration $j$, the verifier picks a random string $r_j$ and create a predicate $P_j(x_{b_1},\ldots, x_{b_q})$, where $b_1,\ldots, b_q$ are the proof bits read by the verifier $V_g^{\Pi}$ on random string $r_j$.  
This predicate is true on assignment $\gamma$ if and only if the verifier accepts the local assignment where $\Pi_{b_i} = \gamma(x_i)$ for all $i \in [q]$.  

First, assume that $\phi$ is satisfiable.
Then there is a proof $\Pi^*$ such that the verifier $V^{\Pi^*}(\phi)$ accepts with probability $\delta$.
Let $\gamma: X \rightarrow \{0,1\}$ be an assignment that agrees with the proof $\Pi^*$.
So $\gamma$ satisfies each predicate $P_j$ with probability $\delta$, and therefore, the expected number of satisfied predicates is $\delta M$.
By Chernoff's bound, the probability that $\gamma$ satisfies less than $\frac{\delta M}{2}$ predicates is at most $2^{-\delta M/ 8} \leq 2^{-n}$.

Next, assume that $\phi$ is not satisfiable.
For each assignment $\gamma: X \rightarrow \{0,1\}$, the fraction of random strings satisfied by the corresponding proof $\Pi_{\gamma}$ is at most $\delta/g$.
When we pick a random string $r_j$, the probability that $V^{\Pi_{\gamma}}(\phi, r_j)$ accepts is then at most $\delta/g$.
So, over all the choices of $M$ strings, the expected number of satisfied predicates is $\delta M/g \geq 10 S(n)$.
By Chernoff's bound, the probability that $\gamma$ satisfies more than $\delta M/g$ predicates is at most $2^{-10 S(n)}$.
By union bound over all possible proofs of length $S(n)$ (there are $2^{S(n)}$ such proofs), the probability that there is such a $\gamma$ is at most $2^{S(n)} 2^{-10 S(n)} \leq 2^{- S(n)}$.
\end{proof}

\subparagraph*{Step 2: FGLSS reduction}
The FGLSS reduction is a standard reduction from \CSP to \MIS introduced by Feige~et~al.~\cite{FGLSS96}.
The reduction simply lists all possible configurations (partial assignment) for each clause as vertices and adding edges if there is a conflict between two configuration.
In more detail, for each predicate $P_i$ and each partial assignment $\gamma$ such that $P_i(\gamma)$ is true, we have a vertex $v(i,\gamma)$. For each pair of vertices $v(i,\gamma) v(i', \gamma')$ such that there is a variable appearing in both $P_i$ and $P_{i'}$ for which $\gamma(x_j) \neq \gamma'(x_{j})$, we have an edge between $v(i,\gamma)$ and $v(i', \gamma')$.
\begin{lemma}[FGLSS Reduction \cite{FGLSS96}]
There is an algorithm that, given an input \CSP $\phi$ with $m$ clauses, $n$ variables, and freeness $F$, produces a graph $G=(V,E)$ such that (i) $|V(G)| \leq m F$ and (ii) $\alpha(G) = \val(\phi) m$, where $\val(\phi)$ denotes the maximum number of predicates of $\phi$ that can be satisfied by an assignment. 
\end{lemma}

\subparagraph*{Combining everything}
Assume that $\SAT \in \fgpcp_{\delta}(S, F)$.
Let $g >0$ be a constant and $V_g$ be the verifier of \SAT that gives the gap of $g$.
By invoking~Lemma~\ref{lem: pcp to csp}, we have a \CSP $\phi_1$ with $S(n,g)$ variables and $100 S(n,g) g/\delta$ clauses.
Moreover, the freeness and gap of $\phi_1$ are $F(g)$ and $g$ respectively.
Applying the FGLSS reduction, we have a graph $G$ with $N=|V(G)| = 100 S(n,g) F(g) g/\delta= O(S(n,g) F(g) g)$.
Now assume that we have an algorithm ${\mathcal A}$ that gives a $g$ approximation in time $2^{\frac{o(S^{-1}(N,g))}{g F(g)}}$.
Notice that $S^{-1}(N,g) \leq O(n g F(g))$ and therefore algorithm ${\mathcal A}$  distinguishes between Yes- and No-instance in time $2^{o(n)}$, a contradiction.

\vspace{-0.1in}

\subparagraph*{Hardness under Gap-ETH:} Dinur~\cite{Dinur16} and  Manurangsi and Raghavendra~\cite{ManurangsiR16}  made a conjecture that \SAT does not admit an approximation scheme that runs in $2^{o(n)}$ time.
We observe a Gap-ETH  hardness of $r$-approximating \MIS in time $2^{n/r^c}$ for some constant $c$.
The proof uses a standard amplification technique and is deferred to Appendix~\ref{app:gapeth}.

\section{Further Research}
Our work leaves ample opportunity for exciting research. An obvious open question is to derandomize our branching, e.g. whether Theorem~\ref{thm:main} can be proved without randomized algorithms. While the probabilistic approximation guarantee can be easily derandomized by using a random partition of the vertex set in $r$ parts or splitters, it seems harder to strengthen the expected running time bound to a worst-case running time bound.

Can we improve the running times of the other algorithms mentioned in the introduction that use the partition argument, possibly using the randomized branching strategy? Specifically, can we $(1+\eps)$-approximate \MIS on planar graphs in time $O^*(2^{(1/\eps)/\log(1/\eps)})$, or $r$-approximate \MIS in time $O^*(2^{tw/r \log r})$? As mentioned in the introduction, a result of Marx~\cite{Marx07} still leaves room for such lower order improvements. 
Another open question in this category is how fast we can $r$-approximate $k$-\MIS, where the goal is to find an independent st of size of $k$. For example no $O(n^{k/f(r)})$ time algorithm is known, where $f(r)$ is a non-trivial function of $r$, that distinguishes graphs $G$ with $\alpha(G) \geq 2k$ from graphs with $\alpha(G) \leq k$. The partition argument gives only a running time of $(n/r)^{0.792k}$, and no strong lower bounds are known for this problem.
Finally, a big open question in the area is to find or exclude a $(2-\eps)$-approximation for \VC in graphs in subexponential time for some fixed constant $\eps >0$. 

\vspace{-0.1in} 

\subparagraph*{Acknowledgment} 
NB is supported by a NWO Vidi grant 639.022.211 and ERC consolidator grant 617951. 
BL is supported by ISF Grant No. 621/12 and I-CORE Grant No. 4/11.
DN is supported by the European Research Council (ERC) under the European Union’s Horizon 2020 research and innovation programme under grant agreement No 715672 and the Swedish Research Council (Reg. No. 2015-04659).
JN is supported by NWO Veni grant 639.021.438.

\bibliography{refs}

\begin{thebibliography}{10}

\bibitem{DBLP:journals/jacm/Baker94}
Brenda~S. Baker.
\newblock Approximation algorithms for np-complete problems on planar graphs.
\newblock {\em J. {ACM}}, 41(1):153--180, 1994.

\bibitem{BGG15}
Nikhil Bansal, Anupam Gupta, and Guru Guruganesh.
\newblock On the {L}ov{\'{a}}sz {T}heta {F}unction for {I}ndependent {S}ets in
  {S}parse {G}raphs.
\newblock In {\em Symposium on Theory of Computing, {STOC}}, pages 193--200,
  2015.

\bibitem{BellareGS98}
Mihir Bellare, Oded Goldreich, and Madhu Sudan.
\newblock Free bits, pcps, and nonapproximability-towards tight results.
\newblock {\em {SIAM} J. Comput.}, 27(3):804--915, 1998.

\bibitem{BHK09}
Andreas Bj{\"{o}}rklund, Thore Husfeldt, and Mikko Koivisto.
\newblock Set partitioning via inclusion-exclusion.
\newblock {\em {SIAM} J. Comput.}, 39(2):546--563, 2009.

\bibitem{BLP16}
{\'{E}}douard Bonnet, Michael Lampis, and Vangelis~Th. Paschos.
\newblock Time-approximation trade-offs for inapproximable problems.
\newblock In {\em Symposium on Theoretical Aspects of Computer Science,
  {STACS}}, pages 22:1--22:14, 2016.

\bibitem{BP16}
{\'E}douard Bonnet and Vangelis~Th. Paschos.
\newblock Sparsification and subexponential approximation.
\newblock {\em Acta Informatica}, pages 1--15, 2016.

\bibitem{BEP11}
Nicolas Bourgeois, Bruno Escoffier, and Vangelis~Th. Paschos.
\newblock Approximation of max independent set, min vertex cover and related
  problems by moderately exponential algorithms.
\newblock {\em Discrete Applied Mathematics}, 159(17):1954 -- 1970, 2011.

\bibitem{CIP06}
Chris Calabro, Russell Impagliazzo, and Ramamohan Paturi.
\newblock A duality between clause width and clause density for {SAT}.
\newblock In {\em Conference on Computational Complexity {(CCC)}}, pages
  252--260, 2006.

\bibitem{CLN13}
Parinya Chalermsook, Bundit Laekhanukit, and Danupon Nanongkai.
\newblock Independent set, induced matching, and pricing: Connections and tight
  (subexponential time) approximation hardnesses.
\newblock In {\em Foundations of Computer Science, {FOCS}}, pages 370--379,
  2013.

\bibitem{Chan16}
Siu~On Chan.
\newblock Approximation resistance from pairwise-independent subgroups.
\newblock {\em J. {ACM}}, 63(3):27:1--27:32, 2016.

\bibitem{CKPW08}
Marek Cygan, Lukasz Kowalik, Marcin Pilipczuk, and Mateusz Wykurz.
\newblock Exponential-time approximation of hard problems.
\newblock {\em CoRR}, abs/0810.4934, 2008.

\bibitem{CyganKW09}
Marek Cygan, Lukasz Kowalik, and Mateusz Wykurz.
\newblock Exponential-time approximation of weighted set cover.
\newblock {\em Inf. Process. Lett.}, 109(16):957--961, 2009.

\bibitem{CyganP10}
Marek Cygan and Marcin Pilipczuk.
\newblock Exact and approximate bandwidth.
\newblock {\em Theor. Comput. Sci.}, 411(40-42):3701--3713, 2010.

\bibitem{CHLN05}
Artur Czumaj, Magn{\'{u}}s~M. Halld{\'{o}}rsson, Andrzej Lingas, and Johan
  Nilsson.
\newblock Approximation algorithms for optimization problems in graphs with
  superlogarithmic treewidth.
\newblock {\em Inf. Process. Lett.}, 94(2):49--53, 2005.

\bibitem{Dinur07}
Irit Dinur.
\newblock The {PCP} theorem by gap amplification.
\newblock {\em J. {ACM}}, 54(3):12, 2007.

\bibitem{Dinur16}
Irit Dinur.
\newblock Mildly exponential reduction from gap 3sat to polynomial-gap
  label-cover.
\newblock {\em Electronic Colloquium on Computational Complexity {(ECCC)}},
  23:128, 2016.

\bibitem{Feige04}
Uriel Feige.
\newblock Approximating maximum clique by removing subgraphs.
\newblock {\em {SIAM} J. Discrete Math.}, 18(2):219--225, 2004.

\bibitem{FGLSS96}
Uriel Feige, Shafi Goldwasser, L{\'{a}}szl{\'{o}} Lov{\'{a}}sz, Shmuel Safra,
  and Mario Szegedy.
\newblock Interactive proofs and the hardness of approximating cliques.
\newblock {\em J. {ACM}}, 43(2):268--292, 1996.

\bibitem{FeigeK98}
Uriel Feige and Joe Kilian.
\newblock Zero knowledge and the chromatic number.
\newblock {\em J. Comput. Syst. Sci.}, 57(2):187--199, 1998.

\bibitem{Halperin02}
Eran Halperin.
\newblock Improved approximation algorithms for the vertex cover problem in
  graphs and hypergraphs.
\newblock {\em {SIAM} J. Comput.}, 31(5):1608--1623, 2002.

\bibitem{Hastad96}
Johan H{\aa}stad.
\newblock Clique is hard to approximate within n\({}^{\mbox{1-epsilon}}\).
\newblock In {\em 37th Annual Symposium on Foundations of Computer Science,
  {FOCS}}, pages 627--636, 1996.

\bibitem{IPZ01}
Russell Impagliazzo, Ramamohan Paturi, and Francis Zane.
\newblock Which problems have strongly exponential complexity?
\newblock {\em J. Comput. Syst. Sci.}, 63(4):512--530, 2001.

\bibitem{Khot16}
Subhash Khot, Dor Minzer, and Muli Safra.
\newblock On independent sets, 2-to-2 games and grassmann graphs.
\newblock {\em Electronic Colloquium on Computational Complexity {(ECCC)}},
  23:124, 2016.

\bibitem{KhotP06}
Subhash Khot and Ashok~Kumar Ponnuswami.
\newblock Better inapproximability results for maxclique, chromatic number and
  min-3lin-deletion.
\newblock In {\em Automata, Languages and Programming, International
  Colloquium, {(ICALP)}}, pages 226--237, 2006.

\bibitem{KhotR08}
Subhash Khot and Oded Regev.
\newblock Vertex cover might be hard to approximate to within 2-epsilon.
\newblock {\em J. Comput. Syst. Sci.}, 74(3):335--349, 2008.

\bibitem{KS16}
Subhash Khot and Igor Shinkar.
\newblock On hardness of approximating the parameterized clique problem.
\newblock In {\em Innovations in Theoretical Computer Science {(ITCS)}}, pages
  37--45, New York, NY, USA, 2016. ACM.
\newblock \href {http://dx.doi.org/10.1145/2840728.2840733}
  {\path{doi:10.1145/2840728.2840733}}.

\bibitem{Laekhanukit2014-Thesis}
Bundit Laekhanukit.
\newblock {\em Inapproximability of Combinatorial Problems in
  Subexponential-Time}.
\newblock PhD thesis, McGill University, 2014.

\bibitem{ManurangsiR16}
Pasin Manurangsi and Prasad Raghavendra.
\newblock A birthday repetition theorem and complexity of approximating dense
  csps.
\newblock {\em CoRR}, abs/1607.02986, 2016.

\bibitem{Marx07}
D{\'{a}}niel Marx.
\newblock On the optimality of planar and geometric approximation schemes.
\newblock In {\em Foundations of Computer Science {(FOCS)}}, pages 338--348,
  2007.

\bibitem{MoshkovitzR10}
Dana Moshkovitz and Ran Raz.
\newblock Two-query {PCP} with subconstant error.
\newblock {\em J. {ACM}}, 57(5):29:1--29:29, 2010.

\bibitem{NP85}
Jaroslav Nešetřil and Svatopluk Poljak.
\newblock On the complexity of the subgraph problem.
\newblock {\em Commentationes Mathematicae Universitatis Carolinae},
  026(2):415--419, 1985.

\bibitem{WY}
Ryan Williams and Huacheng Yu.
\newblock Personal communication.

\end{thebibliography}

\clearpage
\appendix

\section{A Deterministic Algorithm for \MIS}\label{app:feige}

In this section, we give a deterministic $r$-approximation algorithm that runs in time $2^{O(n/ r \log r)}$.
This algorithm is a simple consequence of Feige's algorithm~\cite{Feige04}, that we restate below in a slightly different form.

\begin{theorem}[\cite{Feige04}]
Let $G$ be a graph with independence ratio $\frac{\alpha(G)}{|V(G)|} = 1/k$.
Then, for any parameter $t$, one can find an independent set of size $\Omega(t \cdot \log_k (\frac{n}{kt}))$ in time $\poly(n) k^{O(t)}$.
\end{theorem}

Now, our algorithm proceeds as follows.

\begin{itemize}
\item If $\alpha(G) < n/\log^2 r$, we can enumerate all independent sets of size $n/(r \log^2 r)$ (this is an $r$-approximation) in time ${n \choose n/(r \log^2 r)} \leq (e r \log^2 r)^{\frac{n}{r \log^2 r}} \leq 2^{O(n/(r \log r))}$.

\item Otherwise, the independence ratio is at least $1/k$ where $k = \log^2 r$.
We choose $t = n/(r \log r)$, so Feige's algorithm finds an independent set of size at least
$$
\Omega\left(t \cdot \log_k (\frac{n}{kt})\right)
= \Omega\left(\frac{n}{ r\log r} \cdot \log_k (r \log r)\right)
= \Omega(n/(r \log \log r))
$$
The running time is
$$k^{O(t)} = 2^{O(\frac{n(\log \log r)}{r \log r})}$$
If we redefine $r'=  r \log \log r$, then the algorithm is an $r'$-approximation algorithm that runs in time $2^{O(n (\log \log r')^2/r' \log r')}$.
\end{itemize}

\section{Gap-ETH hardness of \MIS (sketch)} \label{app:gapeth}
We now sketch the proof. 
We are given an $n$-variable \ThreeSAT formula $\phi$ with perfect completeness and soundness $1-\epsilon$ for some $\epsilon >0$.
We first perform standard amplification and sparsification to get $\phi'$ with gap parameter $g$, the number of clauses is $n g$, and freeness is $g^{O(1/\epsilon)}$.
Now, we perform FGLSS reduction to get a graph $G$ such that $|V(G)| = n g^{O(1/\epsilon)}$.
Therefore, $g$-approximation in time $2^{o(|V(G)|/g^{O(1/\epsilon)})}$ would lead to an algorithm that satisfies more than $(1-\epsilon)$ fraction of clauses in \ThreeSAT formula in time $2^{o(n)}$.
In other words, any $2^{n/r^c}$-time algorithm that $r$-approximates \MIS can be turned into a $(1+O(1/c))$-approximation algorithm for approximating \ThreeSAT in sub-exponential time. 

\end{document}